\documentclass[prd,reprint, nofootinbib,amsmath,amssymb, aps, floatfix]{revtex4-1}

\usepackage{subfigure}
\usepackage{amssymb,amsmath}
\usepackage{graphicx, xcolor, varwidth}
\usepackage{hyperref}
\usepackage{braket}
\usepackage{mathrsfs}
\usepackage{amsthm}
\newtheorem{thm}{Theorem}
\newtheorem{assump}{Assumption}

\newtheorem{cond}{Condition}
\newtheorem{definition}{Definition}
 
\newtheorem{lem}[thm]{Lemma}
\theoremstyle{remark}

\usepackage{bbold}
\usepackage{comment}

\begin{document}
\title{Constraints on the resolution of spacetime singularities}
\author{Arvin Shahbazi-Moghaddam}
\affiliation{Leinweber Institute for Theoretical Physics, Stanford University, Stanford, CA 94305}
\affiliation{Leinweber Institute for Theoretical Physics, University of California,
 Berkeley, CA 94720}

\begin{abstract}

What happens at spacetime singularities is poorly understood. The Penrose–Wall singularity theorem constrains possible scenarios, but until recently its key assumption—the generalized second law (GSL)—had only been proven perturbatively, severely limiting this application. We highlight that recent progress enables a proof of the GSL in holographic brane-world models, valid non-perturbatively at the \emph{species scale} $cG$ (with $c$ the number of matter fields and $G$ Newton’s constant). This enables genuine constraints: an outer-trapped surface in the Einstein gravity regime implies geodesic incompleteness \emph{non-perturbatively} at the species scale. Conversely, any genuine resolution must evade Penrose’s criteria. We illustrate both possibilities with explicit examples: the classical BTZ black hole evolves to a more severe singularity, while a null singularity on the Rindler horizon is resolved, both by species-scale effects. Subject to the GSL, these constraints on singularity resolution apply beyond brane-worlds: namely, in any theory with a \emph{geometric UV scale}—roughly, where the metric remains well-defined but classical Einstein gravity breaks down.

\end{abstract}

\maketitle

\section{Introduction}

A robust signature of a spacetime singularity is the presence of an incomplete causal geodesic, i.e., a future (or past) inextendible causal curve with finite affine parameter~\cite{Wald}. For example, in the Schwarzschild spacetime, causal geodesics falling into the $r=0$ singularity are incomplete (see Fig.~\ref{fig:Sch}). A major advantage of \emph{defining} singularities via the presence of geodesic incompleteness is that the latter can be proved via singularity theorems. This is extremely powerful because the conditions of the theorems may be verified in the low curvature region, without the (impractical) need to solve for the full geometry. The Penrose singularity theorem is a primary example, and will be the main focus in this paper~\cite{Penrose:1964wq}.

Penrose's singularity theorem establishes that in a spacetime with non-compact Cauchy slices, the existence of an outer-trapped surface (i.e., a compact spacelike codimension-two surface whose outward null expansion is everywhere negative) implies the existence of at least one incomplete null geodesic. The theorem put to rest the idea that singularities are unphysical artifacts of highly non-generic solutions, at least in the regime of validity of the theorem, that is, classical Einstein gravity coupled to matter satisfying the null energy condition (NEC).\footnote{The NEC states that $T_{ij}k^{i}k^{j}\geq0$, where $T_{ij}$ is the stress-energy tensor and $k^{i}$ is any null vector. In fact, the regime of validity of Penrose's theorem is larger, and includes any spacetime which satisfies the null curvature condition, i.e., $R_{ij}k^i k^j\geq0$, where $R_{ij}$ is the Ricci tensor. The conditions are equivalent in Einstein gravity.}

This is a powerful result on the status of classical singularities. It is believed, however, that Einstein gravity is a truncation of a more complete theory of quantum gravity. And this truncation provides a bad approximation in large curvature regions where Ultraviolet (UV) effects (e.g., quantum effects of matter or gravity) become important. For instance, the NEC can be violated boundlessly in the quantum regime~\cite{Ford:1978qya}. This challenges the physical relevance of Penrose's theorem. One must either give it up when studying quantum gravity, or extend it to a broader regime of validity.

\begin{figure}[t]
    \centering
    \includegraphics[width=.9\columnwidth]{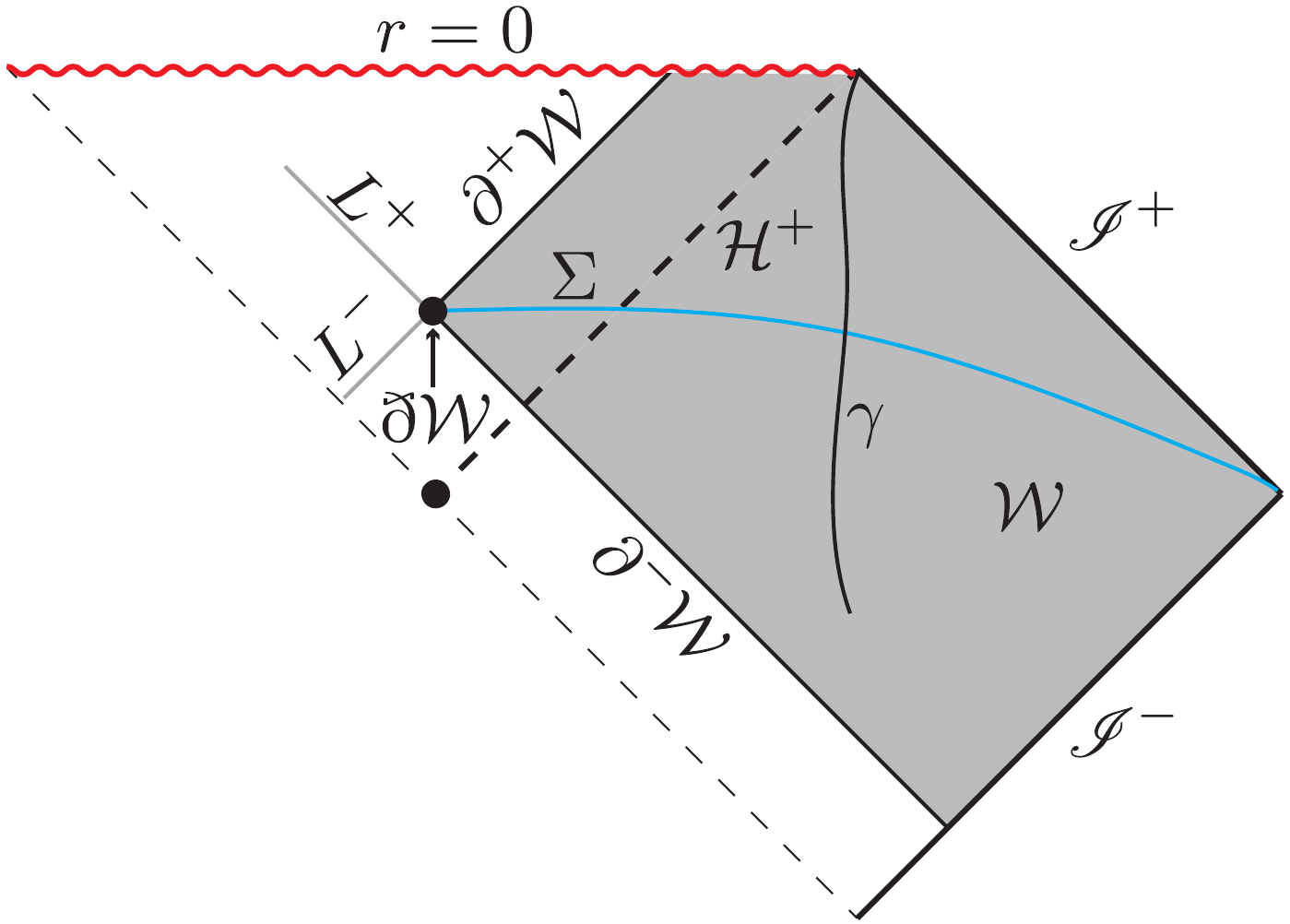}
    \caption{In the Schwarzschild spacetime, a spacetime wedge $\mathcal{W}$ is shown in grey, with a non-compact Cauchy slice $\Sigma$, and with future/past boundaries $\partial^{\pm}\mathcal{W}$. $L^\pm$ denotes portions of the boundaries of the future/past of $\mathcal{W}$. The generators of $\partial^+\mathcal{W}$ are incomplete, since they fall into the $r=0$ singularity after a finite affine parameter. A causal horizon $\mathcal{H}^+$ can be defined as $\partial I^-(\gamma)$, the boundary of the past of the future-infinite causal curve $\gamma$.}
    \label{fig:Sch}
\end{figure}

In~\cite{Wall:2010jtc}, significant progress was made in the latter direction, which can be described as follows: First, it was shown that the NEC assumption can be replaced by the (weaker) assumption of a horizon area law, namely, that cross-sectional areas of future (resp. past) causal horizons cannot decrease towards the future (resp. past)~\cite{HawEll}.\footnote{A future causal horizon is defined as the boundary of the past of any future-infinite causal curve. A past horizon can be defined in the analogous way. In a theory like Einstein gravity which enjoys time-reversal symmetry, the future and past area laws are equivalent.} This is a powerful deduction, since this feature of causal horizons is believed to be the classical gravity manifestation of a more general truth in quantum gravity. In particular, it was conjectured that it generalizes to the so-called generalized second law (GSL) in the presence of quantum effects~\cite{Bekenstein:1973ur}. Roughly speaking, the GSL states that the generalized entropy ($S_{\rm gen}$), a sort of entropy of spacetime regions, cannot decrease along causal horizons. The generalized entropy of a spacetime wedge $\mathcal{W}$ (i.e., a domain of dependence of a partial Cauchy slice) is schematically:
\begin{align}\label{eq-Sgen}
    S_{\rm gen}(\mathcal{W}) = \frac{A(\eth \mathcal{W})}{4G} + S_{\rm ren}(\mathcal{W}),
\end{align}
where $\eth \mathcal{W}$ denotes the edge of the wedge $\mathcal{W}$, $G$ denotes the renormalized Newton's constant, and $S_{\rm ren}$ denotes the renormalized von Neumann entropy of quantum fields inside $\mathcal{W}$~\cite{Bekenstein:1973ur, Susskind:1994sm}.\footnote{A widely believed feature of the generalized entropy is that it is well-defined despite the fact that matter von Neumann entropy in subregions is ill-defined. A schematic explanation for this is that the UV divergence of the von Neumann entropy renormalizes Newton's constant $G$ in the area term, resulting in a finite expression \eqref{eq-Sgen}~\cite{Susskind:1994sm}. See~\cite{Witten:2021unn} for a more modern perspective.} For instance, Hawking radiation causes the area of black hole horizons to shrink, though their generalized entropies grow due to compensation from the von Neumann entropy of the radiation, upholding the GSL~\cite{Page:1976df,Page:1983ug}.

In~\cite{C:2013uza}, it was then shown that the GSL implies a quantum generalization of Penrose's singularity theorem whereby in a spacetime with non-compact Cauchy slices, a compact (quantum) trapped surface, i.e., a spacelike codimension-two surface whose associated outward null (quantum) expansion is negative, implies the existence of a singularity (see Sec.~\ref{sec:review} for a review). We will henceforth refer to this as the Penrose-Wall singularity theorem (PW). The theorem unlocks several potential applications. We focus here on its implications regarding the resolution of spacetime singularities.

A major puzzle in quantum gravity is finding mechanisms for singularity resolution. Near a curvature singularity, one expects large corrections to classical Einstein gravity that can affect the geometry significantly. It is then a logical possibility that what appears as a singularity in a solution of Einstein gravity gives way to a region of large but finite curvature in a more complete theory. In the language of geodesic incompleteness, the Einstein gravity truncation of the theory can, in principle, predict a metric with geodesic incompleteness, even though the exact metric is geodesically complete. An application of the PW singularity theorem would rule out this scenario for singularities which are diagnosed by a (quantum) trapped surface. This does not mean that the singularity must fundamentally persist in the ultimate theory of quantum gravity, but that at least a more exotic mechanism of singularity resolution would have to occur. This would be a powerful conclusion; however, as we now explain, the statement is subject to important caveats.

The caveats we focus on concern the regime of validity of the theorem's central assumption, namely, the GSL. Until recently, the GSL was only rigorously proven in the perturbative quantum gravity regime~\cite{Wall:2011hj}.\footnote{The proof involves reducing the non-trivial case to perturbations around a stationary horizon, in which case it was argued that the GSL is equivalent to the monotonicity of relative entropy in quantum field theory on the stationary horizon, which holds completely generally.} In this regime, we start with a $d$-dimensional classical spacetime $g^{(0)}_{ij}$ on which quantum fields (including gravitons) can propagate. In the limit $G \to 0$, one then corrects the metric order-by-order in $G$ perturbation theory (schematically):
\begin{align}\label{eq-pertquant1}
    \langle g_{ij}\rangle = g^{(0)}_{ij} + \text{$\frac{\ell_P}{L}$ perturbative corrections}
\end{align}
where $\ell_P \equiv G^{\frac{1}{d-2}}$ is the Planck length scale, and $L$ denotes collectively all relevant physical scales, e.g., the curvature radius of the geometry $g^{(0)}_{ij}$, and length scales associated to the state. Due to gravitons, the metric is an operator, and the LHS of Eq.~\eqref{eq-pertquant1} must be made sense of as an expectation value.

There is a fundamental obstruction in analyzing singularity resolution within any such perturbative regime. Let us demonstrate this with a basic example.\footnote{We thank Edward Witten for emphasizing this obstruction with an example of this sort. See also footnote 21 in~\cite{C:2013uza}.} Suppose, for instance, that in a solution to Einstein gravity, the Ricci scalar $R$ as a function of proper time $t$ to the singularity is given by $R = t^{-2}$. In particular, there is a singularity at $t=0$. Now, suppose $\ell_P$ corrections change the function to $R = (t^{2}+\ell_P^2)^{-1}$. In particular, $R$ is no longer singular at $t=0$. This resolution would be invisible within the perturbation expansion, where $R = t^{-2} - \ell_P^2 t^{-4} +\cdots$, getting more severe in fact order-by-order. The lesson is that any singularity theorem confined to perturbation theory, may imply the existence of a singularity which is not there non-perturbatively. Therefore, given merely a perturbative GSL proof, one cannot use the PW reliably to constrain scenarios for singularity resolution. 

Naively, this issue would be ameliorated if only we could prove the GSL non-perturbatively in $\ell_P$, extending the regime of applicability of PW beyond perturbation theory. The theorem would then be predicting geodesic incompleteness in the metric that is determined non-perturbatively in $\ell_P$. It is far from clear, however, if the notion of a metric even makes sense at the Planck scale.

To overcome these hurdles in the application of PW, we need two ingredients. First, we need a scale of new physics at which the notion of a metric still makes sense, but it receives large corrections compared to Einstein gravity. Second, we need to prove the GSL \emph{non-perturbatively} at this new scale. We will formalize these ingredients with some definitions below, but first discuss an explicit example which motivates it.

A specific class of models where both of these ingredients are present is the brane-world holographic models (See Sec.~\ref{sec:brane} for a review). These models have a large effective number of matter fields $c$, and can be analyzed in the regime:
\begin{align}\label{eq-regime1}
    G \to 0,~~~~~~\ell_S = (c G)^{\frac{1}{d-2}} = \text{fixed}.
\end{align}
We will refer to $\ell_S$ as the species scale. In the models, the notion of a (classical) metric remains well-defined at length scales $\ell_S$, but its dynamics is no longer governed by a local description, let alone Einstein gravity. This allows for strong corrections to the metric near a singularity derived in the naive Einstein gravity truncation of the theory. Therefore, singularities are in principle resolvable by physics at the $\ell_S$ scale, making it non-trivial to constrain the scenarios in which this can occur using a singularity theorem.

Furthermore, \cite{Shahbazi-Moghaddam:2022hbw} proved a strong gravitational constraint on the evolution of the generalized entropy, called the restricted quantum focusing (rQFC) \emph{non-perturbative} in the species scale. From the rQFC, subject to mild technical assumptions, the GSL follows. This GSL proof is then \emph{non-perturbative} in the species scale. See Appendix~\ref{sec:rqfc} for details.

This explicit example encourages us to believe that the two missing ingredients discussed above might make sense more generally. Let us, then, formalize them by first defining the notion of a geometric UV scale:

\begin{definition}[Geometric UV scale]\label{eq-geoUVscale}
In a theory of gravity, $\ell$ is a geometric UV length scale if the theory has a notion of a spacetime metric $g_{ij}$ non-perturbatively in $\ell$, and such that at large length scales $L$, we can describe $g_{ij}$ with an expansion around $\ell \to 0$:
\begin{align}\label{eq-pert1}
    g_{ij} = g^{(0)}_{ij} + \text{$\frac{\ell}{L}$ perturbative corrections},
\end{align}
where $g^{(0)}_{ij}$ satisfies the Einstein field equation, coupled to classical matter fields that satisfy the NEC.
\end{definition}
For instance, in such a theory, if $g^{(0)}_{ij}$ is the Schwarzschild solution, $g_{ij}$ would then be a metric which is approximately given by $g^{(0)}_{ij}$ everywhere except near $r=0$ where the background curvature scale $L$ shrinks to zero, and therefore we inevitably exit the range of \eqref{eq-pert1}, so the metric $g_{ij}$ receives large corrections from physics at scale $\ell$. In the context of a theory with a geometric UV scale, we will refer to $g^{(0)}_{ij}$ as an Einstein gravity truncation of $g_{ij}$. We will discuss the relationship between $g_{ij}$ and $g^{(0)}_{ij}$ in more detail in subsection~\ref{eq-constrainingidea}.

Besides the brane-world example, a species scale at which effective field theory breaks down, also appears in string theory~\cite{Dvali:2010vm}. It would be interesting to study whether this or other scales directly related to the string length $\sqrt{\alpha'}$ can serve as an instance of a geometric UV scale.

Given a gravity theory with a geometric UV scale, the PW singularity theorem becomes non-trivially applicable if the following assumption holds:
\begin{assump}\label{eq-GSLassume}
In a gravity theory, with a geometric UV scale $\ell$, any causal horizon in the exact metric $g_{ij}$ satisfies the GSL non-perturbatively in $\ell$.
\end{assump}

Then, the PW singularity theorem constrains scenarios for a geometric singularity resolution, defined as follows:
\begin{definition}[Geometric singularity resolution]\label{eq-GeoSingRes}
In a gravity theory with a geometric UV scale $\ell$, suppose $g^{(0)}_{ij}$, an Einstein gravity truncation of the exact metric $g_{ij}$, is geodesically incomplete. We say that the singularity of $g^{(0)}_{ij}$ is \emph{geometrically resolved} if $g_{ij}$ does not contain any incomplete geodesics.
\end{definition}
For instance, in the Schwarzschild example, it is easy to confirm that the conditions of PW hold non-trivially. The theorem then implies that any non-perturbative in $\ell$ metric $g_{ij}$ must be null geodesically incomplete. So a geometric singularity resolution is forbidden. Similar scenarios naturally arise in many cosmological Friedmann–Robertson–Walker metrics $g^{(0)}_{ij}$, where the resolution of a big bang singularity in the exact metric $g_{ij}$ would be forbidden.

\textbf{Relation to previous work:} In~\cite{Bousso:2025xyc}, a different limitation of the PW theorem is emphasized. In particular, the paper analyzes the PW theorem in the regime~\eqref{eq-regime1}, and emphasizes that the so-called ``touching Lemma'' in~\cite{C:2013uza} (specifically, Theorem 1 in~\cite{C:2013uza}) was only proved in a perturbative regime, e.g., order-by-order in a $cG\to0$ expansion in the context of the regime~\eqref{eq-regime1}. A new singularity theorem is then proved that circumvents the Lemma, and hence its perturbative limitation. This is somewhat complementary to the goal of our paper which is, first, to emphasize the inherent limitation of analyzing singularity resolution within perturbation theory (see example under Eq.~\eqref{eq-pertquant1}), and show how this limitation carries over to the PW theorem through the regime of validity of the GSL (as it would to the theorem in~\cite{Bousso:2025xyc}). And to demonstrate that both of these limitations disappear when working within brane-world theories non-perturbatively at species scale. Lastly, to provide explicit examples of how one can use the PW theorem to constrain scenarios for singularity resolution. In our explicit examples below, the touching Lemma is only applied in a regime where it has been proved. Lastly, as an aside, and complementary to the approach of~\cite{Bousso:2025xyc}, we show in Appendix~\ref{sec:rqfc}, that the touching Lemma follows from a cross-focusing condition, proved non-perturbatively at the species scale in the brane-world theories in~\cite{Shahbazi-Moghaddam:2022hbw}, along with a conjecture in~\cite{Bousso:2024iry}.

The outline of the paper is as follows. In Sec.~\ref{sec:review} we review some basic geometric notions, the generalized entropy, and the quantum expansion, building up to the exact statement of the Penrose-Wall singularity theorem. In Subsection~\ref{eq-constrainingidea}, we then explain how within the setting of a theory with a geometric UV scale and the non-perturbative GSL assumption \eqref{eq-GSLassume}, one can use PW as a powerful constraining tool for scenarios of spacetime singularity resolution. To ground these ideas in an example, in Sec.~\ref{sec:brane} we review the holographic brane-world models where a geometric UV scale (the species scale) exists and the GSL was proven to hold non-perturbatively at that scale (see Appendix~\ref{sec:rqfc}). In Sec.~\ref{sec:examples}, we then provide explicit examples of singularity resolution and non-resolution in the brane-world model in $d=3$. In particular, as predictable by PW, the singularity of the (non-rotating) BTZ geometry is shown to persist (and in fact gets more severe) when non-perturbative metric corrections from the species scale are accounted for. We also demonstrate an example of a singularity on the Rindler horizon which \emph{is} resolved by the UV scale. We show that, as required by consistency with PW, this singularity could not have been diagnosed by a trapped surface. We end with a discussion of some future directions in Sec.~\ref{sec:disc}.

\section{Review of the singularity theorem}\label{sec:review}

In this section, we will introduce the basic ingredients of the Penrose-Wall theorem, building up to its statement and a proof sketch. In many places, we diverge from the original presentation of the theorem. For that, we refer the reader to~\cite{C:2013uza}. In subsection~\ref{eq-constrainingidea}, we discuss how in a theory with a geometric UV scale (see Definition~\ref{eq-geoUVscale}), subject to Assumption~\ref{eq-GSLassume}, the PW theorem can be used to meaningfully constrain scenarios for singularity resolution.

\subsection{The geometric ingredients}

Let $(M,g_{ij})$ denote a spacetime manifold $M$ with metric $g_{ij}$. A spacetime wedge $\mathcal{W} \subseteq M$, is the domain of dependence of an achronal codimension-$1$ submanifold $\Sigma$ in $M$ (see Fig.~\ref{fig:Sch}):
\begin{align}
    \mathcal{W} = D(\Sigma).
\end{align}
Then, $\Sigma$ is a Cauchy surface for $\mathcal{W}$. Throughout, we take $\mathcal{W}$ to be an open set.

The boundary of $\mathcal{W}$, denoted by $\partial \mathcal{W}$, has a future and a past piece, defined by:
\begin{align}
    \partial^{+}\mathcal{W} &= \partial\mathcal{W} \cap I^+(\mathcal{W})\\
    \partial^{-}\mathcal{W} &= \partial\mathcal{W} \cap I^-(\mathcal{W}).
\end{align}
where $I^+(S)$ (resp. $I^-(S)$) for some set $S$ denotes the set of points in $M$ which can be reached by a future- (resp. past-) directed timelike curve starting from any point in $S$.

The edge of the wedge, denoted by $\eth \mathcal{W}$, is defined by:
\begin{align}\label{eq-9}
    \eth \mathcal{W} = \partial\mathcal{W} - I(\mathcal{W}),
\end{align}
where $I(S) = I^+(S)\cup I^-(S)$.

The set $\partial^{+}\mathcal{W}$ is a null hypersurface generated by future-directed null geodesics emanating orthogonally from $\eth \mathcal{W}$, and terminated at caustics or self-intersections~\cite{Akers:2017nrr}. Each generator of $\partial^{+}\mathcal{W}$, affinely parameterized with $\lambda=0$ at $\eth \mathcal{W}$, falls into one of the following classes:
\begin{enumerate}
    \item \emph{Future-extendible:} the generator remains on $\partial^{+}\mathcal{W}$ for $\lambda \in [0,\lambda_{\mathrm{end}}]$ with $\lambda_{\mathrm{end}}<\infty$, and then leaves it (at a caustic or self-intersection).
    \item \emph{Future-complete:} the generator has infinite affine extent, i.e.\ it is defined on $[0,\infty)$.
    \item \emph{Future-incomplete:} the generator has finite affine length, i.e.\ it is defined on an interval $[0,\lambda_{*})$ with $0<\lambda_{*}<\infty$ but admits no extension to larger $\lambda$ within $M$.

\end{enumerate}
Topologically, these cases correspond (after affine rescaling) to $[0,1)$, $[0,\infty)$, and $[0,1]$, respectively. These possibilities play an important role both in the Penrose-Wall theorem and in Penrose’s classical singularity theorem.

\subsection{The quantum expansion}

It is believed that in a gravitational theory, one can assign a generalized entropy to any spacetime wedge~\cite{C:2013uza, Myers:2013lva, Bousso:2015mna, Susskind:1994sm}. The generalized entropy is a functional $S_{\rm gen}: \mathcal{W} \to \mathbb{R}$, which in a perturbative regime \eqref{eq-pert1}, is well-approximated by $A(\eth \mathcal{W})/4G$, where $A$ denotes the area.

Both the PW theorem and the GSL are naturally describable in terms of properties of spacetime wedges pertaining to how their generalized entropies vary as we deform them via translating their edges along, say, the boundary of their past (or future). This property has been historically described by a quantum expansion functional $\Theta:\eth\mathcal{W} \to \mathbb{R}$ (see Eq.~\eqref{eq-Thetatheta}). In our presentation, we stick to the modern axiomatic structure developed in~\cite{Bousso:2024iry}, which in particular shows that $\Theta$ can only be ascribed a sign in general, and not a numerical value:\footnote{This is in particular because a numerical value is ill-defined when $\eth\mathcal{W}$ has kinks. Such kinks arise generically in the application of the GSL, singularity theorems, etc.}
\begin{definition}[$\Theta^+\rvert_{\mathcal{W}}(p)<0, \Theta^+\rvert_{\mathcal{W}}(p)\leq 0$]\label{eq-Thetasign}
Let $\mathcal{W}$ be a wedge with edge $\eth \mathcal{W}$, and $L^+ = \partial I^+(\mathcal{W})$, i.e., the boundary of the future of $\mathcal{W}$. We say that $\Theta^+\rvert_{\mathcal{W}}(p)<0$ for a point $p \in \eth \mathcal{W}$, if in any neighborhood of $p$, there exists a future-directed deformation of $\mathcal{W}$ along $L^+$ that decreases $S_{\rm gen}$.

Furthermore, we say that $\Theta^+\rvert_{\mathcal{W}}(p)\leq 0$ for a point $p \in \eth \mathcal{W}$, if there exists a neighborhood of $p$ in which all future-directed deformation of $\mathcal{W}$ along $L^+$ decreases or does not change $S_{\rm gen}$ (see Fig.~\ref{fig:deformed}).
\end{definition}
The definition above can be modified in the obvious way to define $\Theta^-\rvert_{\mathcal{W}}(p)$, by replacing $L^+$ with $L^- = \partial I^-(\mathcal{W})$, the boundary of the past of $\mathcal{W}$ (see Fig.~\ref{fig:Sch}).

We write $\Theta^+\rvert_{\mathcal{W}}<0$ (and similarly for other signs) to indicate that $\Theta^+\rvert_{\mathcal{W}}(p)<0$ for all $p\in\eth\mathcal{W}$.

In any smooth neighborhood of $\eth \mathcal{W}$, the quantum expansion can be assigned a numerical value, i.e., $\Theta^{\pm}\rvert_{\mathcal{W}}:\eth \mathcal{W} \to \mathbb{R}$, via a functional derivative of $S_{\rm gen}$ under null deformations of $\eth \mathcal{W}$ along the generators of $L^\pm$. For instance, take $L^-$. Let $(v,y^a)$ be coordinates on $L^-$, where $\partial_v$ are past-directed affine generators of $L^-$, such that $v=0$ at $\eth \mathcal{W}$, and $y^a$ (with $a=1,\cdots,d-2$) label the transverse directions. Then, each function $v=V(y^a)$ determines a wedge obtained from $\mathcal{W}$ by deforming $\eth\mathcal{W}$ from $v=0$ to $V(y^a)$. The quantum expansion is then:
\begin{align}\label{eq-Qexpansion}
    \Theta^-\rvert_{\mathcal{W}}(y^a) = \frac{4G}{\sqrt{h_V}} \left.\frac{\delta S_{\rm gen}}{\delta V(y^a)}\right\rvert_{V(y^a)=0}
\end{align}
where $h_V$ denotes the determinant of the induced metric of the edge of the wedge $v=V(y^a)$. A numerical value for $\Theta^+\rvert_{\mathcal{W}}(y^a)$ can be assigned in the obvious analogous way.

It is easy to see that for $p$ in a smooth neighborhood of $\eth \mathcal{W}$, the various signs of $\Theta^{\pm}\rvert_{\mathcal{W}}(p)$ according to Definition~\ref{eq-Thetasign} coincide with the corresponding signs of the numerical value.

In a theory with a geometric UV scale, let $\mathcal{W}$ be a wedge with a smooth edge $\eth \mathcal{W}$, and with at least one of its Cauchy slices in the perturbative regime \eqref{eq-pert1}. Then, we can obtain a perturbative expansion for the generalized entropy, with $A(\eth \mathcal{W})/4G$ as its leading term. Therefore, at any smooth point $p\in \eth \mathcal{W}$, we obtain:
\begin{align}\label{eq-Thetatheta}
    \left.\Theta^{\pm}\right\rvert_{\mathcal{W}}(p) = \left.\theta^{\pm}\right\rvert_\mathcal{W}(p) + \text{$\frac{\ell}{L}$ perturbative corrections},
\end{align}
where $\left.\theta^{\pm}\right\rvert_\mathcal{W}$ denotes the classical expansion of the hypersurfaces $L^{\pm}$ at $p$. In the above setting, it is explicitly evaluated as follows:
\begin{align}
    \left.\theta^{-}\right\rvert_\mathcal{W}(p) = \left.\frac{\delta}{\delta V(y^a)}\log \sqrt{h_V}\right\rvert_{V(y)=0}.
\end{align}
where $y^a$ determines $p$. 
\begin{figure}[t]
    \centering
    \includegraphics[width=.7\columnwidth]{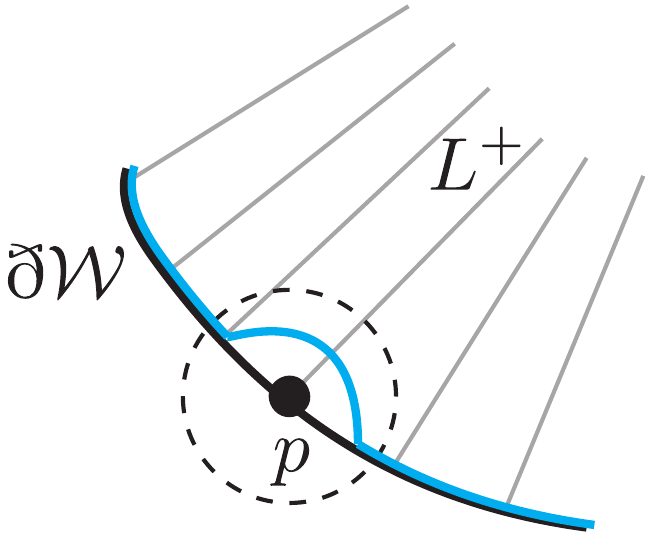}
    \caption{A spacetime wedge can be deformed by moving its $\eth \mathcal{W}$ (from the black line to the blue line) along the boundary of its future $L^+$. If there exists a neighborhood of a point $p \in \eth \mathcal{W}$, inside which any such deformation along $L^+$ decreases the generalized entropy of the wedge or leaves it constant, we say $\Theta^+\rvert_{\mathcal{W}}(p)\leq0$. The sign $\Theta^+\rvert_{\mathcal{W}}(p)<0$ can be defined in a similar way.}
    \label{fig:deformed}
\end{figure}

\subsection{The generalized second law}

We will now discuss the GSL which is the central assumption of the Penrose-Wall singularity theory. Let us first define more precisely the notion of the causal horizon. 

\begin{definition}[Causal Horizon]
In a globally hyperbolic spacetime, given any future-complete (infinite) causal curve $\gamma$, we define a future causal horizon $\mathcal{H}^+$ as the boundary of its past, i.e.,
\begin{align}
    \mathcal{H}^+ \equiv \partial I^-(\gamma).
\end{align}
\end{definition}
The causal horizon $\mathcal{H}^+$ is an achronal null hypersurface. Intuitively, it is the boundary of the region of spacetime that an infinitely long-lived observer sees. This can be a black hole horizon, a cosmological horizon (e.g., in a de Sitter universe), or the Rindler horizon in Minkowski spacetime.

\begin{definition}[Exterior of $\mathcal{H}^+$]\label{eq-exterior}
In a globally hyperbolic spacetime with a future causal horizon $\mathcal{H}^+$, given any (global) Cauchy slice $\Sigma_{\text{global}}$, we call any wedge $\mathcal{W} = D(\Sigma_{\text{global}} \cap I^-(\gamma))$ an exterior of $\mathcal{H}^+$.
\end{definition}

\begin{definition}[Generalized Second Law]\label{eq-GSL}
In a globally hyperbolic spacetime with a future causal horizon $\mathcal{H}^+$, given two exteriors $\mathcal{W}_1$ and $\mathcal{W}_2$ such that $\mathcal{W}_2 \subseteq \mathcal{W}_1$, the GSL states that
\begin{align}
    S_{\rm gen} (\mathcal{W}_2) \geq S_{\rm gen} (\mathcal{W}_1).
\end{align}
\end{definition}
An immediate implication of the GSL is that no exterior wedge $\mathcal{W}$ can satisfy $\Theta^-\rvert_{\mathcal{W}}>0$. In the classical Einstein gravity regime, the GSL reduces to the statement that the cross-sections of $\mathcal{H}^+$ cannot increase in area towards the past.

\subsection{The singularity theorem}

Before stating the PW theorem, we will present a useful lemma. A stronger result was proven within perturbation theory in~\cite{C:2013uza} (Theorem $1$ in the reference).

\begin{lem}\label{eq-tch}
Let $\mathcal{W} \supseteq \tilde{\mathcal{W}}$, such that there is a point $p\in \eth\mathcal{W} \cap \eth\mathcal{\tilde{W}}$. Then, $\Theta^-\rvert_{\mathcal{W}}(p)>0 \implies \Theta^-\rvert_{\tilde{\mathcal{W}}}(p)>0$.
\end{lem}
In classical Einstein gravity, the Lemma~\ref{eq-tch} is a consequence of the fact that, roughly speaking, the boundary of the past of $\mathcal{\tilde{W}}$ is more concave than that of $\mathcal{W}$. (see Fig.~\ref{fig:lemmafig}).

\begin{figure}[t]
    \centering
    \includegraphics[width=.65\columnwidth]{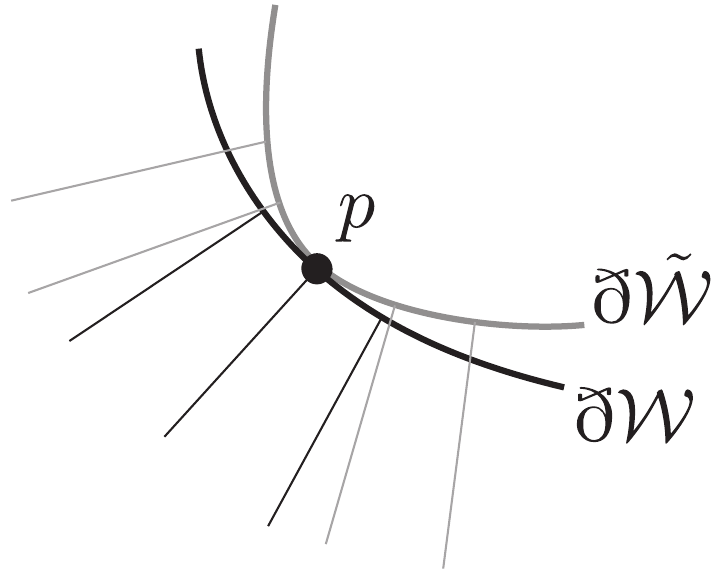}
    \caption{A cross section of time in the scenario discussed in Lemma~\ref{eq-tch}. Given $\mathcal{\tilde{W}}\subseteq \mathcal{W}$, their edges $\eth\mathcal{W}$ and $\eth\mathcal{\tilde{W}}$ intersect at $p$. In such a scenario, the quantum expansion of $\mathcal{W}$ in the towards the boundary of its past (along black rays) is positive, then so is that of $\mathcal{\tilde{W}}$ (along grey rays)}
    \label{fig:lemmafig}
\end{figure}

The proof of Theorem 1 of~\cite{C:2013uza} subsumes Lemma~\ref{eq-tch}. In the context of a geometric UV scale, the proof strategy applies order-by-order in perturbation theory around $\ell\to0$. In~\cite{Bousso:2025xyc}, this perturbative limitation was pointed out, and a modified singularity theorem was proved which circumvents this step. Here, we simply assume Lemma~\ref{eq-tch}, and follow the PW singularity theorem, for two reasons. First, for our purposes of constraining scenarios of singularity resolution, we are content to use Lemma \ref{eq-tch} in the regime where it is proved. This will be explained in more detail in subsection~\ref{eq-constrainingidea}. Second, in Appendix~\ref{sec:rqfc} we take a stride towards proving the Lemma~\ref{eq-tch} non-perturbatively in the species scale of the brane-world theory. This is done by reducing it to a statement which is implied by a strong subadditivity conjecture in~\cite{Bousso:2024iry} (Conjecture 23 in the reference).

Let us finally state the main theorem:

\begin{thm}[Penrose-Wall Singularity Theorem]\label{eq-03c1}
In a spacetime $(M,g_{ij})$, let $\mathcal{W}$ be a spacetime wedge with a compact edge $\eth \mathcal{W}$ and a non-compact Cauchy slice $\Sigma$. Suppose that $\Theta^-|_{\mathcal{W}}>0$. Then, at least one generator of $\partial^+ \mathcal{W}$ is incomplete in $M$.
\end{thm}

\begin{proof}
(Sketch of the proof, see~\cite{C:2013uza} for details) Suppose that no generator of $\partial^+\mathcal{W}$ is incomplete. This can happen (see under Eq.~\eqref{eq-9}) if either all of the generators of $\partial^+\mathcal{W}$ are future-extendible, or at least one of the generators is future-complete (infinite). In the former case, the topology of $\partial^+\mathcal{W}$ is that of $\eth \mathcal{W} \times [0,1]$, which is compact since $\eth \mathcal{W}$ is compact. As shown by Penrose~\cite{Penrose:1964wq}, this contradicts the non-compactness of $\Sigma$, and therefore cannot be the case. Suppose, then, that a generator $\gamma\in\partial^+\mathcal{W}$ is future-complete. Then, $\mathcal{H}^+ = \partial I^-(\gamma)$ is a causal horizon. Let $\mathcal{\tilde{W}} = D(I^-(\gamma) \cap \Sigma)$. It follows from Lemma~\ref{eq-tch}, and the condition $\Theta^-|_{\mathcal{W}}>0$, that $\Theta^-\rvert_{\mathcal{\tilde{W}}}>0$. This is in contradiction with the GSL. Therefore, at least one generator of $\partial^+ \mathcal{W}$ is incomplete.
\end{proof}

The codimension-two surface $\eth \mathcal{W}$ in Theorem~\ref{eq-03c1} was called a quantum (outer) trapped surface in~\cite{C:2013uza}. In the classical limit (i.e., in Einstein gravity coupled to matter satisfying the NEC), where we also get to replace the condition $\Theta^-|_{\mathcal{W}}>0$ with $\theta^-|_{\mathcal{W}}>0$, we obtain the original Penrose singularity theorem~\cite{Penrose:1964wq} for an (outer) trapped surface $\eth \mathcal{W}$.

\subsection{Constraining scenarios for singularity resolution}\label{eq-constrainingidea}

In general, Theorem~\ref{eq-03c1} must be viewed as an exact theorem applicable to the exact metric $g_{ij}$.\footnote{This view is subject to assuming both the GSL and Lemma \ref{eq-tch} non-perturbatively in the geometric UV scale. See~\cite{Bousso:2025xyc}, for an alternative singularity theorem which does not rely on Lemma \ref{eq-tch}.} Such an exact singularity theorem ought to have many interesting applications. Here, we outline a very specific kind of application pertaining to the question of singularity resolution. Specifically, we will explain how theorem~\ref{eq-03c1}, when applied to a gravity theory with a geometric UV scale (see Definition~\ref{eq-geoUVscale}), and satisfying the GSL assumption~\ref{eq-GSLassume}, severely constrains scenarios for a geometric singularity resolution (as in Definition~\ref{eq-GeoSingRes}).

Quite broadly, we are in situations where we know an Einstein gravity truncation of the metric $g^{(0)}_{ij}$, but determining the exact metric $g_{ij}$ is not feasible. Now, suppose $g^{(0)}_{ij}$ contains a singularity. A question of singularity resolution then involves determining what happens to the singularity in the exact metric $g_{ij}$. For instance, we can wonder whether it goes away, stays the same, or change significantly in some way?

Therefore, the starting point in this application is geodesic incompleteness in $g^{(0)}_{ij}$, an Einstein gravity truncation of some exact metric $g_{ij}$. Looking at Eq.~\eqref{eq-pert1}, we learn that the metric $g^{(0)}_{ij}$ can be obtained as an $\ell \to 0$ limit of the exact metric $g_{ij}$. Let us remark here that there is, in general, no unique way to take such a limit.\footnote{We thank Raphael Bousso for discussion on this point.} For example, in the context of the species scale in $d=4$, with $G_4 \to 0$, and $cG_4$ fixed, say the exact metric is prepared by sending in a spherically symmetric shell of matter with mass $M_{\text{total}} = c M_1$, where $M_1$ is the contribution per species. We can take the $cG_4 \to 0$ limit in at least two different ways: first, by sending the dimensionless ratio $cG_4 / A \to 0$, where $A \sim c^2 G_4^2 M_1^2$ is the area of the black hole. The resulting $g^{(0)}_{ij}$ is a black hole formed by collapse, approaching the Schwarzschild metric in the future and such that the area of the black hole is infinitely large compared to the species scale, making the perturbative analysis in $cG_4$ a good approximation. Second, we can send the dimensionless quantity $cG_4 M_1^2 \to 0$ . The result is flat space with a shell of total mass $M_{\text{total}}$ which does not backreact on the geometry.

The non-uniqueness of $g^{(0)}_{ij}$ is not an obstruction to our discussion of singularity resolution, since the procedure we outline below works for any $g^{(0)}_{ij}$, which is an Einstein gravity truncation of $g_{ij}$, i.e., results from some $\ell \to 0$ limit of an exact metric $g_{ij}$. Suppose such a $g^{(0)}_{ij}$ contains a wedge $\mathcal{W}^{(0)}$ with a non-compact Cauchy slice $\Sigma^{(0)}$, and a compact $\eth \mathcal{W}^{(0)}$ satisfying $\theta^-\rvert_{\mathcal{W}^{(0)}} >0$. Then, by Penrose's classical theorem, $\partial^+\mathcal{W}^{(0)}$ has an incomplete null geodesic in $g^{(0)}_{ij}$. This is a singularity whose resolution we can constrain. 

Suppose that $\Sigma^{(0)}$ is fully contained in a neighborhood of $g^{(0)}_{ij}$ with upper-bounded magnitude of curvature and stress-energy tensor. Since $g^{(0)}_{ij}$ is an $\ell \to 0$ limit of an exact metric $g_{ij}$, any such $\Sigma^{(0)}$ lives in the perturbative regime~\eqref{eq-pert1}. Therefore, $\Sigma^{(0)}$ must have an uplift to a partial Cauchy slice of the exact metric $g_{ij}$, retaining its topological properties (since they are stable under small perturbations). That is, there must exist a partial Cauchy slice $\Sigma$ in spacetime $g_{ij}$ with an $\ell \to 0$ limit that recovers $\Sigma^{(0)}$ in $g^{(0)}_{ij}$. Let us now define $\mathcal{W} = D(\Sigma)$. From Eq.~\eqref{eq-Thetatheta}, it follows that $\Theta^-|_{\mathcal{W}}>0$. The wedge $\mathcal{W}$ then satisfies all of the conditions of the PW singularity theorem. Intuitively, by only investigating the low-curvature region, and without having to solve for or even know the detailed dynamics of the exact theory, we diagnose that a geodesic incompleteness present in $g^{(0)}_{ij}$ cannot be geometrically resolved in the exact metric $g_{ij}$. Equivalently, in any example where $g_{ij}$ does geometrically resolve a singularity in $g^{(0)}_{ij}$, the above diagnostic must be evaded.

Note that since $\mathcal{W}$ has a Cauchy slice in the perturbative regime, so does $\tilde{\mathcal{W}}$ in the proof of Theorem~\ref{eq-03c1}, and therefore in this particular application we believe that Lemma~\ref{eq-tch} can be applied reliably based on its perturbative proof in~\cite{C:2013uza}. Nevertheless, see Appendix~\ref{sec:rqfc} for a non-perturbative argument for Lemma~\ref{eq-tch}.

Next, we will ground the above ideas in a concrete model of gravity, the so-called holographic brane-world models. In Sec.~\ref{sec:examples}, we then discuss three examples, involving singularity resolution and non-resolution, and show how this is compatible with the above story of constraining singularity resolution using the PW theorem.

\section{Holographic brane-world Models}\label{sec:brane}

Here, we review the basics of holographic brane-world models, where we can showcase the power of the above-mentioned ideas. The theories live on a $d-$dimensional brane, and are (holographically) dual to a $(d+1)-$dimensional ``bulk'', governed by Einstein gravity with a negative cosmological constant, in which the brane is embedded. At low curvature, the brane-world theory is approximately a local-on-the-brane theory of gravity (dominated by Einstein gravity) coupled to a holographic conformal field theory (CFT).  At a species scale, this description breaks down and is replaced by a non-local one from the brane point of view. Nevertheless, the brane intrinsic metric remains well-defined, and its dynamics are governed by the local $(d+1)$-dimensional bulk Einstein gravity description. This is an explicit example of a theory with a geometric UV scale (See Definition~\eqref{eq-geoUVscale}).

 Here, we will describe these models from a bottom-up approach. For a more detailed analysis, including some top-down constructions, we refer the reader to~\cite{Randall:1999vf, Karch:2000ct, Myers:2013lva, Gubser:1999vj,Verlinde:1999xm,Geng:2020fxl}.

In the classical bulk regime (i.e., $G_{d+1}\to 0$), standard AdS$_{d+1}$/CFT$_d$ instructs us to compute the CFT partition function on a metric $g_{ij}$ in terms of the on-shell Einstein gravity action of the bulk with Dirichlet boundary conditions $g_{ij}$ on an infrared cutoff surface. The cutoff surface is then sent to infinity (after proper counter-terms are added) and one obtains the renormalized CFT partition function. In the brane-world model, one instead integrates over the metric $g_{ij}$. In effect, this turns Dirichlet into Neumann boundary conditions and turns the fictitious cut-off surface into a dynamical brane. The (Euclidean) partition function of the brane-world theory can be computed holographically:
\begin{align}\label{eq-partition}
    \log Z_{\text{brane-world}} = -I_{\text{bulk}}
\end{align}
where $I_{\text{bulk}}[g_{ij}]$ is the on-shell action of the bulk geometry together with an end-of-the-world brane with tension $T$. Explicitly,
\begin{align}\label{eq-action1}
    I_{\text{bulk}} &= \frac{1}{16\pi G_{d+1}}\int_{\text{bulk}} dzd^{d}x\,\sqrt{\bar{g}}\left(\bar{R} + \frac{d(d-1)}{\ell_{\text{AdS}}^2}\right) \nonumber\\
    &\quad + \frac{1}{8\pi G_{d+1}}\int_{\text{brane}} d^dx\,\sqrt{g}\,\left(K - 8 \pi G_{d+1} T\right),
\end{align}
where the barred quantities $\bar{g}$ and $\bar{R}$ denote the bulk metric determinant and the bulk Ricci scalar respectively. Also, $g = \det(g_{ij})$, where $g_{ij}$ is the brane's induced metric, and $K = g^{ij} K_{ij}$, where $K_{ij}$ is the extrinsic curvature tensor (with respect to the normal pointing into the bulk).

\begin{figure}[t]
    \centering
    \includegraphics[width=.65\columnwidth]{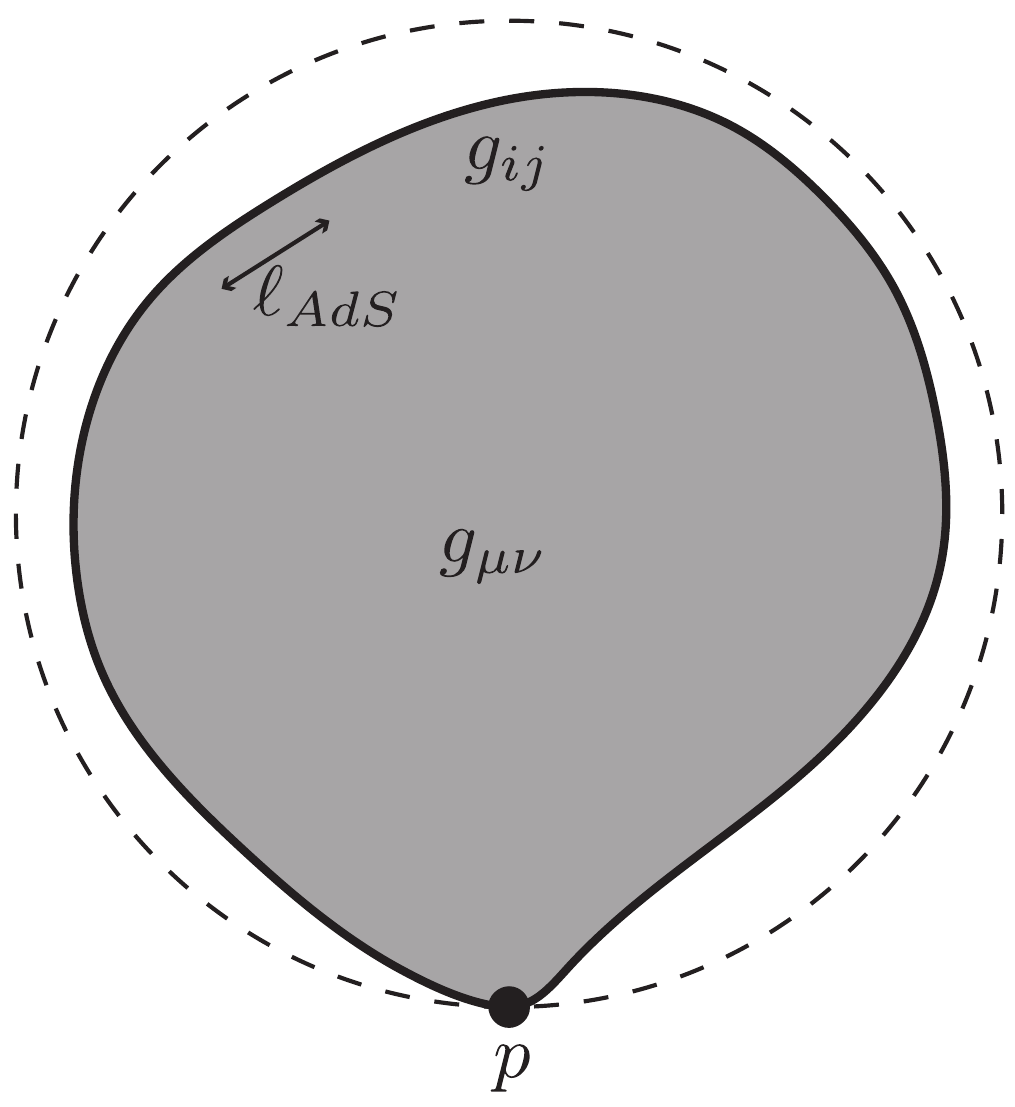}
    \caption{A cross section of Lorentzian time in a holographic brane-world theory (living on the solid black curve) with metric $g_{ij}$, and its bulk dual (shaded grey) with metric $g_{\mu\nu}$. The dashed line indicates the would-be asymptotic boundary of AdS if the bulk was not terminated at the brane. Dual to a classical Einstein gravity bulk (with negative cosmological constant), the brane theory at length scales much larger than $\ell_{\rm AdS}$ is dominated by Einstein gravity (in one lower dimension), but also includes higher curvature and quantum matter corrections. At length scales comparable to $\ell_{\rm AdS}$, this description breaks down, but the intrinsic metric $g_{ij}$ is still computable. Hence, the brane-world theory has a geometric UV scale.}
    \label{fig:brane}
\end{figure}

The location of the brane is dynamically determined by:
\begin{align}\label{eq-braneembedd}
    K_{ij} - g_{ij} K = -8 \pi G_{d+1} T g_{ij}
\end{align}
where $g_{ij}$ and $K_{ij}$ are the brane induced metric and its extrinsic curvature tensor (with respect to the normal pointing into the bulk).

Let us now explain in what sense these theories contain a geometric UV scale. In a given bulk solution with the end-of-the-world brane solution above, let $L$ denote the curvature length scale on the brane. It is possible to show that for $L \gg \ell_{\rm AdS}$, we can approximate the brane-world partition function \eqref{eq-partition} with~\cite{Myers:2013lva}:
\begin{align}\label{eq-actionappr}
      \log Z_{\text{brane-world}} &= \frac{1}{16\pi G_d}\int d^dx\,\sqrt{g}\,\Bigl(-2\Lambda + R \nonumber\\
      &+\ell_{\text{AdS}}^2 (\alpha_1 R^2+\alpha_2 R_{ij}R^{ij}+\alpha_3 R_{ijkl}R^{ijkl})\nonumber\\
      &+\cdots \Bigr)+ \log Z_{\text{CFT}}[g_{ij}] + \cdots,
\end{align}
where $G_d \equiv (d-2)G_{d+1}/\ell_{\rm AdS}$ denotes the effective Newton's constant of the brane theory, and $\alpha_n$ are O(1) Wilson coefficients. The brane cosmological constant $\lambda$ can be tuned by dialing the brane tension $T$. The case $\Lambda=0$ is often referred to as the Randall-Sundrum brane~\cite{Randall:1999vf}, and the case of $\Lambda<0$ is referred to as the Karch-Randall brane~\cite{Karch:2000ct}. Eq.~\eqref{eq-actionappr} shows that when $L\gg\ell_{\rm AdS}$, the brane-world theory is governed by Einstein gravity, plus higher curvature corrections suppressed by appropriate powers of $\ell_{\rm AdS}$. At large enough order in $\ell_{\rm AdS}$, we also obtain a renormalized effective action of the CFT which lives on the brane.

The theory therefore contains both quantum matter effects, and higher curvature corrections to Einstein gravity. But we see that in the $L \gg \ell_{\rm AdS}$ regime, the metric $g_{ij}$ on the brane is governed approximately by Einstein gravity on the brane:
\begin{align}
    g_{ij} = g^{(0)}_{ij} + \text{$\frac{\ell_{\rm AdS}}{L}$ perturbative expansion}
\end{align}
However, even for $L \lesssim \ell_{\rm AdS}$, there still exists a notion of the induced metric on the brane. But the approximate (local) effective theory \eqref{eq-actionappr} is no longer a useful description, due to strong higher curvature and quantum matter effects on the geometry. One must directly solve for the embedding of the brane in the higher dimensional bulk using Eq.~\eqref{eq-braneembedd} to obtain the exact $g_{ij}$. This requires solving a higher dimensional Einstein gravity problem involving a brane whose intrinsic curvature is comparable to the AdS scale. Therefore, we see that $\ell_{\rm AdS}$, which is a mundane scale in the $(d+1)$-dimensional description, becomes a scale of new physics on the brane-world at which the brane theory is not even approximately governed by Einstein gravity.

The scale $\ell_{\rm AdS}$ has a direct interpretation on the brane as a species scale. To see this, let us recall the standard AdS/CFT relationship:
\begin{align}\label{eq-c}
    c \sim \frac{\ell_{\rm AdS}^{d-1}}{G_{d+1}}
\end{align}
where $c$ roughly measures the number of degrees of freedom on the CFT.\footnote{More precisely, in known examples of AdS/CFT in various dimensions, $c$ is the stress-energy tensor $2$-point function normalization.} We can then combine Eq.~\eqref{eq-c} with $G_d = (d-2)G_{d+1}/\ell_{\rm AdS}$ to obtain:
\begin{align}
    \ell_{\rm AdS} \sim (c G_d)^{\frac{1}{d-2}}
\end{align}
where $(c G_d)^{\frac{1}{d-2}}$ is the brane's species scale, which we will henceforth denote by $\ell_S$. Since the bulk is in the limit $G_{d+1}\to0$, with $\ell_{\rm AdS}$ fixed, on the brane-world this means:
\begin{align}
    G_d\to0, ~~~~~~\ell_S \equiv (cG_d)^{\frac{1}{d-2}} = \text{fixed}
\end{align}

The species scale has been discussed extensively in the literature as a scale where the local Einstein gravity description breaks down~\cite{Dvali:2010vm}. Here, we see it explicitly emerge as a natural geometric UV scale in the brane-world scenario.

Not only do holographic brane-world theories have a geometric UV scale, they also satisfy a strong gravitational constraint called restricted quantum focusing, non-perturbatively at their species scale~\cite{Shahbazi-Moghaddam:2022hbw}. The non-perturbative GSL follows from the rQFC subject to very mild technical assumptions. We review this in Appendix~\ref{sec:rqfc}. We therefore have all of the ingredients for a non-trivial constraining of singularity resolution scenarios using the PW singularity theorem. In the next section, we give explicit examples of singularity resolution and non-resolution and its consistency with the PW theorem.

\section{Explicit examples}\label{sec:examples}

In practice, finding exact analytic solutions to the brane-world scenario is tricky, though they can in principle be found numerically as they merely require solving some non-linear partial differential equations associated with Einstein gravity in the bulk and the brane embedding condition~\eqref{eq-braneembedd}. See for example~\cite{Figueras:2011gd, Biggs:2021iqw, Almheiri:2019psy} for numerical constructions. Several exact solutions are known, however, three of which we provide here to showcase the ideas in this paper. The examples are in $d=3$ and are all borrowed from references~\cite{Emparan:1999wa, Emparan:2002px, Emparan:2020znc}. They involve a non-trivial solution known as the AdS$_4$ $C$-metric in which an end-of-the-world brane can be embedded~\cite{PLEBANSKI197698}. In the examples below, we simply discuss the resulting metric on the brane. It would be interesting to find other non-trivial examples in $d\geq3$.

\subsection{Non-rotating BTZ black hole}

In~\cite{Emparan:2020znc}, it was discovered that the BTZ solution is an Einstein gravity truncation of an exact solution in the brane-world scenario (specifically, the Karch-Randall brane). These are found on a brane-world theory with action \eqref{eq-actionappr} satisfying $d=3$, and $\lambda =-1/\ell_3^2$.~\footnote{Furthermore, $\alpha_1=3/8$, $\alpha_2=-1$, $\alpha_3=0$~\cite{Emparan:2020znc}.}

Let us first analyze the Einstein gravity truncated solution. Recall that in $d=3$ Einstein gravity with a negative cosmological constant, the BTZ solution is given by the line element:
\begin{align}\label{eq-BTZ}
       (ds^2)^{(0)} = &-\left( \frac{r^2}{\ell_3^2} - 8 G_3 M \right) dt^2+\frac{dr^2}{\frac{r^2}{\ell_3^2} - 8 G_3 M } + r^2 d\phi^2,
\end{align}
where we have the identification $\phi \sim \phi+2\pi$. Here, $M$ is the black hole mass.

Let us then take a connected spherically symmetric partial Cauchy slice $\Sigma^{(0)}$ located at some fixed time $r$ in the interior (recall that $r$ is a time coordinate in the black hole interior). See Fig.~\ref{fig:threefigs} a. The resulting wedge $\mathcal{W}^{(0)} = D(\Sigma^{(0)})$ clearly satisfies all of the conditions of Penrose's theorem. In particular, $\theta^-\rvert_{\mathcal{W}^{(0)}}>0$, because interior transverse areas decrease towards $r=0$.

Indeed, this geometry contains a Lorentzian conical singularity at $r=0$, and the causal geodesics that hit $r=0$ are incomplete. In the brane-world theory, where matter fields are present, the expectation value of the matter stress-energy tensor diverges at $r=0$~\cite{Emparan:2020znc}. We therefore expect large corrections in this region in the exact metric. This in turn presents a possibility that the singularity in the exact metric would go away.

But, as $\Sigma^{(0)}$ is the low curvature region, it must have an uplift to a partial Cauchy slice $\Sigma$ in the exact metric $g_{\mu\nu}$ and such that $\mathcal{W} = D(\Sigma)$ satisfies the conditions of the Penrose-Wall theorem. Since the GSL holds non-perturbatively in the species scale $\ell_S$, PW reliably forbids that the large corrections near $r=0$ resolve the singularity, and predicts that geodesic incompleteness persists. 

Indeed, the exact metric $g_{ij}$ is given by:
\begin{align}\label{eq-qBTZ}
    ds^2 = &-\left( \frac{r^2}{\tilde{\ell}_3^2} - 8 G_3 \tilde{M} - \frac{\ell_S H(\tilde{M})}{r}  \right) dt^2\nonumber\\
    &+\frac{dr^2}{\frac{r^2}{\tilde{\ell}_3^2} - 8 G_3 \tilde{M} - \frac{\ell_S H(\tilde{M})}{r}} + r^2 d\phi^2
\end{align}
Here, $\tilde{M}$, and $\tilde{\ell}_3$ equal $M$, and $\ell_3$ respectively to leading order in $\ell_S$. Furthermore, $H(\tilde{M})$ is a certain smooth function of $\tilde{M}$ whose details do not matter for our purposes.

Therefore, the exact metric~\eqref{eq-qBTZ} is even more severely singular at $r=0$ compared to $g^{(0)}_{ij}$. For instance, $R_{ij} R^{ij} \sim r^{-6}$ near $r=0$. Note that the metric $\eqref{eq-qBTZ}$ is non-perturbatively different from its Einstein gravity truncation~\eqref{eq-BTZ}. Nevertheless, without the need to solve for it explicitly, we could correctly predict its failure to resolve the singularity geometrically.
\begin{figure*}[t]
    \centering
    \includegraphics[width=0.8\textwidth]{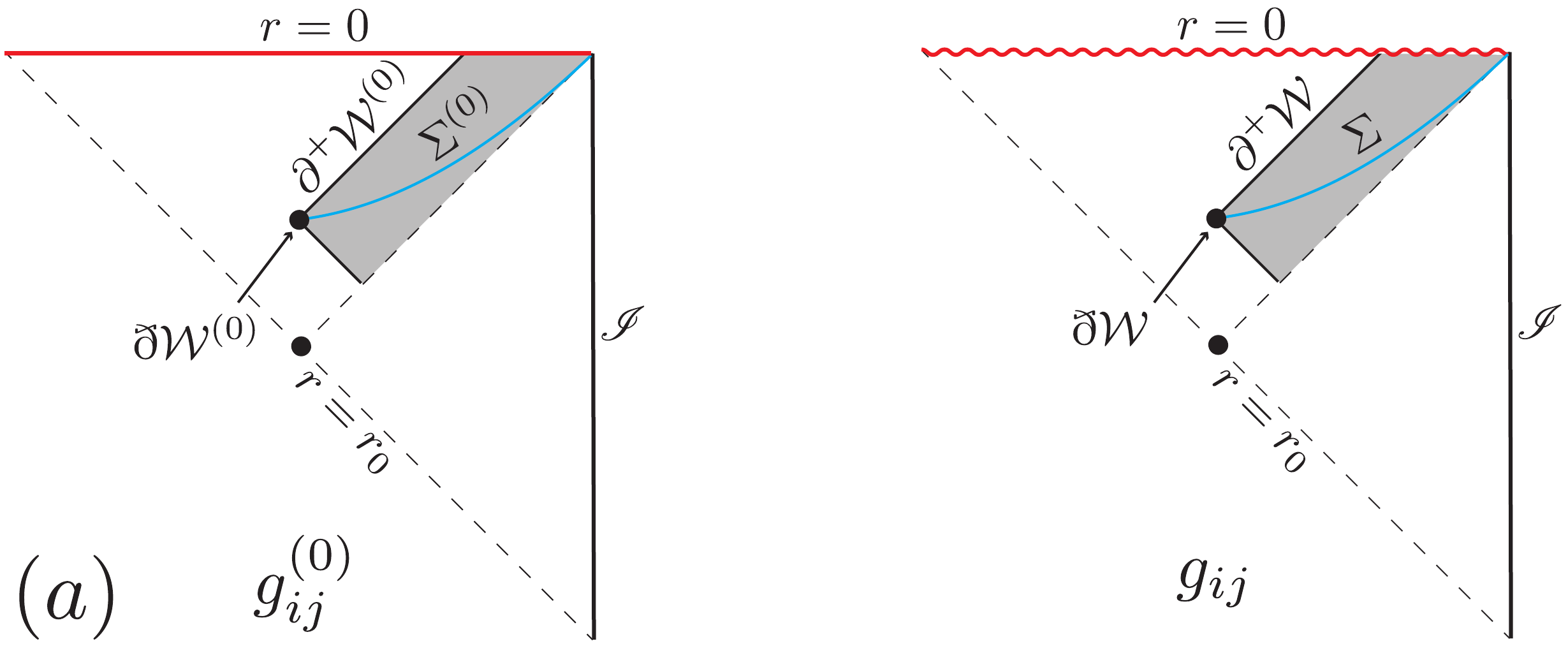}
    \includegraphics[width=0.8\textwidth]{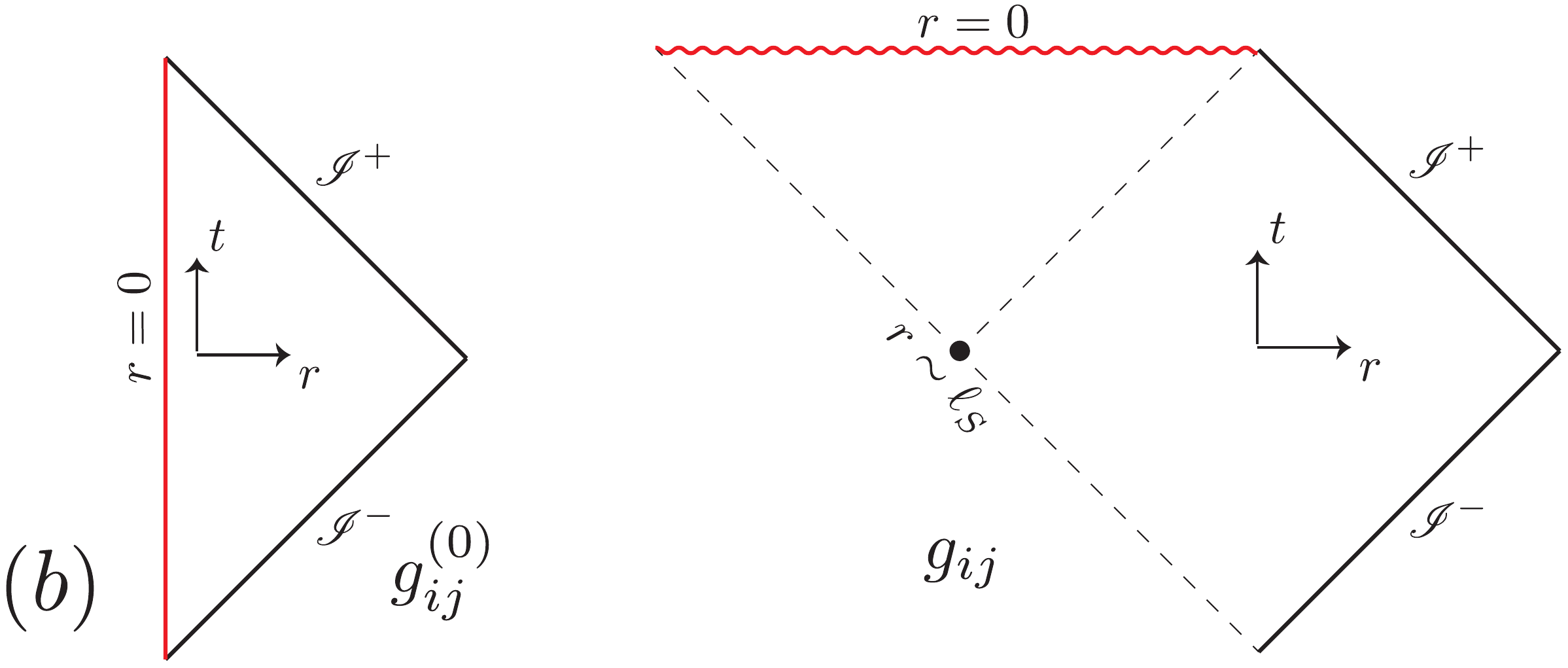}
    \includegraphics[width=0.8\textwidth]{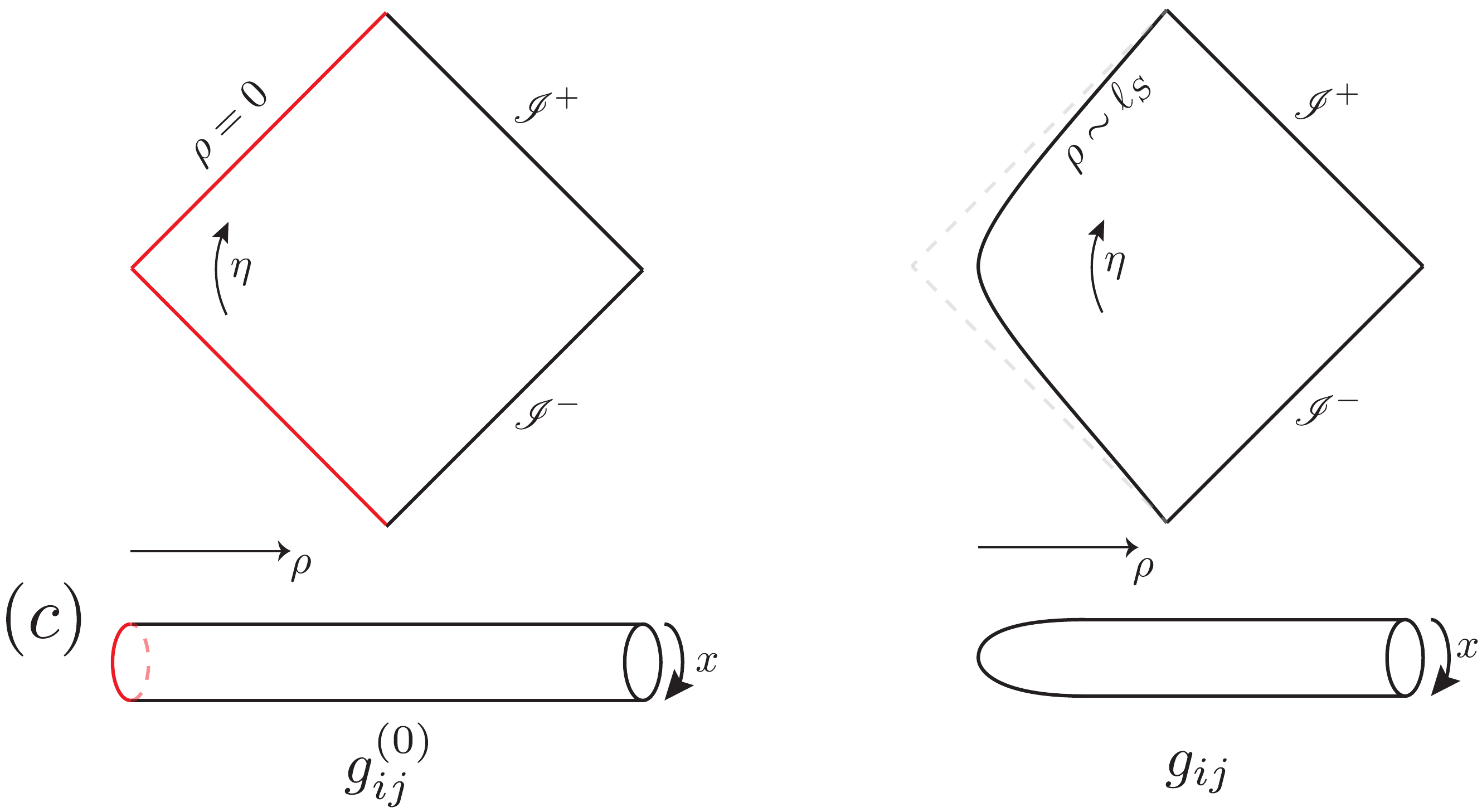}
    \caption{Examples of Sec.~\ref{sec:examples}. Red lines (straight or wiggly) show the loci of geodesic incompleteness: (a) left: BTZ black hole in Einstein gravity. The wedge $D(\Sigma^{(0)})$ satisfies the conditions of PW non-trivially. The generators of $\partial^+\mathcal{W}^{(0)}$ are incomplete. (a) right: In the non-perturbative uplift of the spacetime, the geodesic incompleteness persists as predicted by PW. The conical singularity of BTZ becomes a curvature singularity due to species-scale corrections. (b) left: A conical singularity at $r=0$ in Minkowski spacetime. The quantum fields are divergent at $r=0$ signaling large corrections once $\ell_S$ effects are turned on (b) right: At finite $\ell_S$, the spacetime turns into a species-scale black hole with a spacelike singularity behind its horizon. (c) left: A Rindler wedge with a small Kaluza-Klein circle, with locus of incompleteness on its horizon. (c) right: The singularity is geometrically resolved by species scale physics, as the Kaluza-Klein circle caps off near the would-be singularity at $\rho$ of order $\ell_S$.}
    \label{fig:threefigs}
\end{figure*}
\subsection{Conical singularity in Minkowski spacetime}

The next example is that of a conical deficit singularity in $d=3$ Minkowski spacetime, which can also be obtained as an Einstein gravity truncation of an exact brane-world solution~\cite{Emparan:1999wa}. Recall the line element of a conical deficit:
\begin{align}\label{eq-cone0}
    (ds^2)^{(0)} = -dt^2 + dr^2 + r^2d\phi^2, ~~~~~\phi \sim \phi +2\pi \alpha,
\end{align}
where $0<\alpha\leq 1$. See Fig.~\ref{fig:threefigs} b. For $\alpha<1$, the spacetime has a conical singularity at $r=0$. It then makes sense to remove $r=0$ from the manifold, which means geodesics hitting $r=0$ are incomplete. Again, in the presence of matter fields, the quantum stress-energy tensor $\langle T_{ij} \rangle$ diverges at $r=0$, and therefore the exact solution must receive large corrections near $r=0$. The exact brane-world solution from which Eq.~\eqref{eq-cone0} descends from (as $\ell_S \to 0$) is~\cite{Emparan:1999wa}:
\begin{align}\label{eq-2c14c5}
ds^2 = -\left(1-\frac{r_0}{r}\right)dt^2 + \frac{dr^2}{\left(1-\frac{r_0}{r}\right)} + r^2 d\phi^2,
\end{align}
with
\begin{align}\label{eq-2c14c6}
    r_0 = \ell_S~f(\alpha),
\end{align}
where $f(\alpha)$ is a smooth monotonic function which satisfies $f(\alpha=1)=0$, and $f(\alpha\to0)\to\infty$. The geodesics hitting $r=0$ are incomplete.

The metric \eqref{eq-2c14c5} is that of a $d=3$ black hole solution whose emblackening factor resembles a Schwarzschild metric. This is therefore not a vacuum Einstein gravity solution in $d=3$. Indeed, as we see in Eqs.~\eqref{eq-2c14c5} and \eqref{eq-2c14c6}, the horizon size is of order $\ell_S$, and therefore the physics is not governed by $d=3$ Einstein gravity, which only becomes a good approximation at length scales much larger than $\ell_S$.


This is a scenario where both the Einstein gravity truncation \eqref{eq-cone0} and the exact solution \eqref{eq-2c14c5} are null geodesically incomplete. However, the character of the singularity is quite different. In \eqref{eq-cone0}, we have a timelike singularity in a non-globally hyperbolic spacetime, whereas, after the $\ell_S$ corrections are exactly accounted for, we obtain a small black hole horizon with a spacelike singularity (see Fig.~\ref{fig:threefigs}b).

Here, the specific application explained in subsection~\ref{eq-constrainingidea} does not apply, since the spacetime \eqref{eq-cone0} does not have a wedge to apply the PW theorem to non-trivially. Note that, as we explain in the beginning of subsection~\ref{eq-constrainingidea}, the PW theorem is still non-trivially applicable directly to the exact metric $g_{ij}$: there are (species-scale sized) quantum trapped surfaces inside the black hole \eqref{eq-2c14c5} which imply the geodesic incompleteness at $r=0$.


\subsection{Singular Rindler wedge}

Our last example is one in which a geometric singularity resolution happens. It is obtained by a double analytic continuation of the solution \eqref{eq-2c14c5} and, to our knowledge, has not appeared elsewhere. First, starting from the Lorentzian conical deficit \eqref{eq-cone0}, we send $t \to i x$, and simply re-label $r\to \rho$. The resulting metric is:
\begin{align}
    (ds^2)^{(0)} = \rho^2 d\phi^2+d\rho^2+ dx^2, ~~~~\phi\sim\phi+2\pi\alpha
\end{align}
We recognize this metric as Euclidean space with a conical deficit at $\rho=0$. Euclidean conical deficits induce divergences of $\langle T_{ij} \rangle$. We can then open up to a different Lorentzian section by sending $\phi\to i \eta$ to obtain:
\begin{align}\label{eq-Rindler}
    (ds^2)^{(0)} = -\rho^2 d\eta^2+d\rho^2+ dx^2,~~~x \sim x+ \Delta.
\end{align}
with $\Delta = 4 \pi r_0$. The origin of the $\Delta$ periodicity is that of Euclidean time in the black hole solution. See Fig.~\ref{fig:threefigs} c.

The spacetime in \eqref{eq-Rindler} with $\rho\geq0$ is the Rindler wedge. The conical deficit in its Euclidean preparation means that the quantum fields are at Rindler temperature $2\pi \alpha$ which for $\alpha\neq1$ are not in the Minkowski vacuum sector. In particular, they have $\langle T_{ij} \rangle$ singularities at $\rho=0$, hinting that the exact metric would be quite different there. Furthermore, the spacetime is geodesically incomplete because lightrays can reach the Rindler horizon at finite affine parameter.

The exact brane-world solution is given by
\begin{align}\label{eq-resolved}
ds^2 = -\rho^2 d\eta^2 +\frac{d\rho^2}{\left(1-\frac{\rho_0}{\rho}\right)} +\left(1-\frac{\rho_0}{\rho}\right)dx^2, ~~~x \sim x+ \Delta,
\end{align}
with $\rho_0=\ell_S~f(\alpha)$. The spacetime \eqref{eq-resolved} is geodesically complete! The mechanism by which the singularity is cured is that the Kaluza-Klein direction $x$ shrinks towards $\rho\to0$ and caps off at some $\rho\sim\ell_S$. This replaces the $\rho=0$ region with divergent $\langle T_{ij} \rangle$ with a smooth but highly curved metric, i.e., species-scale curvature.

Consistency with the PW singularity theorem requires that the Rindler wedge \eqref{eq-Rindler} not have any partial Cauchy slice $\Sigma^{(0)}$ in the bounded $L$ region satisfying the conditions of Penrose's theorem. The boundedness condition prohibits us from considering a $\Sigma^{(0)}$ which touches $\rho=0$. It is easy to investigate that no other choice of $\Sigma^{(0)}$ would also lead to a contradiction.


\section{Future Directions}\label{sec:disc}

It would be interesting to investigate the full range of applicability of the above ideas. For instance, does the existence of a trapped surface inside of a Schwarzschild black hole constrain the mechanism of singularity resolution in string theory?

Furthermore, we expect that all of the lessons in this paper also apply to a singularity theorem discovered in~\cite{Bousso:2022tdb}. It would be interesting to constrain singularity resolution scenarios using this theorem as well.

Lastly, we expect that a lot more explicit examples of singularity resolution and non-resolution can be explored in brane-world theories. It would be interesting to construct more analytic or numerical examples to learn more about how the species scale changes the nature of a singularity in Einstein gravity.

\section*{Acknowledgments}

I am grateful to Alek Bedroya, Raphael Bousso, Juan Maldacena, Aron Wall, Edward Witten, and Yoav Zigdon for helpful discussions. I especially thank Raphael Bousso for detailed comments on a draft of this paper.

This work is supported by the LITP at UC Berkeley, Department of Energy through DE-SC0019380, and DE-FOA0002563, by AFOSR award FA9550-22-1-0098 and by a Sloan Fellowship.

\appendix

\section{rQFC implies the GSL}\label{sec:rqfc}

Here, in subsection~\ref{sub1}, we review how the GSL follows from the rQFC on smooth portions of null hypersurfaces, subject to a mild assumption (Assumption~\ref{eq-generic}). We believe that the assumption can either be relaxed, or proved, in the context of brane-world theories with more work along the lines of~\cite{Shahbazi-Moghaddam:2022hbw}. In subsection~\ref{sub2}, we show that using a cross-focusing relation discussed in~\cite{Shahbazi-Moghaddam:2022hbw}, Lemma~\ref{eq-tch} can be reduced to the case where $\eth \mathcal{W}$ and $\eth\mathcal{\tilde{W}}$ coincide in a neighborhood of their touching point $p$. From there, the Lemma follows from the strong subadditivity conjecture (23) of \cite{Bousso:2024iry}.\footnote{The conjecture has so far not been proven exactly, say at the species scale.}

The discussion here in principle applies to any theory where $S_{\rm gen}$ can be defined exactly for any spacetime wedge. But while the \emph{statements} of the rQFC and cross-focusing do not depend on any specific such theory, they have only been proven non-perturbatively in brane-world theories~\cite{Shahbazi-Moghaddam:2022hbw}. Let us therefore remark briefly on the proof strategy used in~\cite{Shahbazi-Moghaddam:2022hbw}, and why it is non-perturbative at the species scale. 

In brane-world theories, given any wedge $\mathcal{W}$ within the brane geometry, one defines its generalized entropy as~\cite{Emparan:2006ni, Myers:2013lva}
\begin{align}\label{eq-Sgenbrane}
    S_{gen}(\mathcal{W}) = \frac{A(X)}{4G_{d+1}},
\end{align}
where $X$ denotes the minimal \emph{bulk} extremal surface anchored to $\eth\mathcal{W}$ on the brane.\footnote{Here, it is assumed that $X$ exists.} For regions on the brane that are large compared to the species scale $\ell_S$ (equivalently, $\ell_{\rm AdS}$), Eq.~\eqref{eq-Sgenbrane} can be expanded order-by-order in $\ell_S$, with the leading term given by $A(\eth \mathcal{W})/ 4G_d$, along with infinitely many sub-leading terms. In particular, this expansion includes the (renormalized) von Neumann entropy of the matter fields in $\mathcal{W}$. For small regions, this expansion breaks down and Eq.~\eqref{eq-Sgenbrane} is the only way to define $S_{\rm gen}$.

From Eq.~\eqref{eq-Sgenbrane}, one can in particular define quantum expansions according to Eq.\eqref{eq-Qexpansion} for smooth portions of $\eth \mathcal{W}$. Both the rQFC and the cross-focusing relation discussed below are proved using properties of the extremal surface deviation equation in the bulk (and assuming the bulk NEC). Since Eq.~\eqref{eq-Sgenbrane} is exact in the species scale $\ell_S$, so are the proofs.

\subsection{Non-perturbative proof of the GSL}\label{sub1}

Let us present a gravitational constraint from which the GSL follows. We later argue that this constraint itself follows from the rQFC assuming a mild genericity condition.

\begin{cond}\label{eq-condcond}
Let $\mathcal{W}_0$ be any wedge, and let $L^-$ denote its boundary of the past (the obvious analogue holds for the boundary of the future). Let $\mathcal{W}_1$ and $\mathcal{W}_2$ be two wedges which can be obtained from $\mathcal{W}_0$ by deforming $\eth \mathcal{W}_0$ along $L^-$. Then:
\begin{align}\label{eq-cond1}
    \Theta^-\rvert_{\mathcal{W}_0} \leq 0 \implies S_{\rm gen}(\mathcal{W}_2) \leq S_{\rm gen}(\mathcal{W}_1),~~\text{for $\mathcal{W}_1 \subseteq \mathcal{W}_2$}
\end{align}
\end{cond}
The obvious analogue condition (similarly derivable from the rQFC) can be stated for $L^+$, the boundary of the future of $\mathcal{W}_0$, and where we replace the LHS of~\eqref{eq-cond1} with $\Theta^+\rvert_{\mathcal{W}_0}$.

The condition $\mathcal{W}_1 \subseteq \mathcal{W}_2$ in \eqref{eq-cond1} enforces that $\mathcal{W}_2$ is further along $L^-$ than $\mathcal{W}_1$. In the literature, Condition~\eqref{eq-cond1} is a special case of Conjecture (33) of an axiomatic framework of quantum gravity outlined in~\cite{Bousso:2024iry}, which, roughly speaking, demands Condition~\ref{eq-condcond} piecewise in the transverse direction. As we will show later, Condition~\eqref{eq-condcond} follows from the rQFC given a genericity assumption.\footnote{In fact, rQFC implies the full Conjecture 33 in~\cite{Bousso:2024iry} subject to this genericity assumption.}

Let us first demonstrate how the GSL follows from Condition\ref{eq-condcond}. By definition, any causal horizon $\mathcal{H}^+$ extends to an asymptotic region. We assume that in this asymptotic region, physics is governed to a good approximation by classical Einstein gravity. In particular, the quantum expansion of the exterior regions of $\mathcal{H}^+$ is well-approximated by the classical expansion of $\mathcal{H}^+$ which is non-negative towards the future by the classical area law, which we assume.

Now, let $\mathcal{W}_1$ and $\mathcal{W}_2$ be two exteriors of $\mathcal{H}^+$ such that $\mathcal{W}_2 \subseteq\mathcal{W}_1$. Let $\mathcal{W}_0 \supseteq \mathcal{W}_2$ be an exterior of $\mathcal{H}^+$ in its asymptotic region. Then, the GSL (see Definition~\ref{eq-GSL}) is an immediate implication of Condition \ref{eq-condcond}, applied to $\mathcal{W}_0$, since $\Theta^-\rvert_{\mathcal{W}_0} \leq 0$ by the asymptotic region area law.

We will now state the rQFC. Let $\mathcal{W}$ be a spacetime wedge, and $L^-$ the boundary of its past.\footnote{The rQFC also applies in the obvious way to $L^+$, the boundary of the future of $\mathcal{W}$, and also to $\partial^\pm \mathcal{W}$.} In a smooth neighborhood of $\eth \mathcal{W}$, let $(v,y^a)$ be coordinates on $L^-$ such that $\partial_v$ are its affine generators pointing towards the past, and such that $v=0$ is $\eth \mathcal{W}$, and $y^a$ denote the transverse direction on $L^-$. In that neighborhood, any function $v = V(y)\geq0$ denotes a wedge obtained from $\mathcal{W}$ by deforming its edge along $L^-$. We henceforth denote any such wedge by the function $V(y^a)$. Let $V_\lambda(y^a)$ be a one-parameter family (parameterized by $\lambda$) of such wedges. The rQFC states the following~\cite{Shahbazi-Moghaddam:2022hbw}:
\begin{cond}\label{eq-rQFCcond}
    Given any $V_\lambda(y^a)$ satisfying $\partial_\lambda V_\lambda(y^a)\geq0,~\forall y^a$, we have:
\begin{align}\label{eq-rQFC}
    \partial_\lambda \Theta^-_{\lambda} (y^a)\leq0,~ \text{for any $y^a,\lambda$ such that $\Theta^-_{\lambda}(y^a)=0$,}
\end{align}
where $\Theta^-_\lambda \equiv \Theta^-\rvert_{\mathcal{W_\lambda}}$, and $\mathcal{W}_\lambda$ denotes the wedge associated to $V_\lambda(y^a)$.
\end{cond}

Let us now state the following genericity assumption:

\begin{assump}[Genericity]\label{eq-generic}
For any $\Theta^-_\lambda (y^a)$ defined as above, we assume that for any $y^a$, $\Theta_\lambda(y^a)$ is a smooth function of $\lambda$. Furthermore, if there exists a $\lambda^*$ such that
\begin{align}
   \Theta_{\lambda^*} (y_0^a) = \left.\partial_{\lambda} \Theta_\lambda\right\rvert_{\lambda=\lambda^*}(y_0^a)=0
\end{align}
for some $y_0^a$, then $\Theta_\lambda (y_0^a)=0$ in a neighborhood of $\lambda^*$.
\end{assump}

We can now show how Eq.~\eqref{eq-cond1} follows from Eq.~\eqref{eq-rQFC} subject to this genericity assumption. First, Eq.~\eqref{eq-cond1} is equivalent to:
\begin{align}\label{eq-cond1implication}
    \Theta^-_{\mathcal{W}}\leq0 \implies \Theta^-_{V(y^a)}\leq0
\end{align}
where we are denoting by $V(y^a)$ the wedge obtained by deforming $\mathcal{W}$ along $L^-$. The implication from Eq.~\eqref{eq-cond1} to Eq.~\eqref{eq-cond1implication} is obvious. For the other direction, we need to, roughly speaking, integrate \eqref{eq-cond1implication} between the two wedges in Eq.~\eqref{eq-cond1}. We leave this as an exercise.

Now, suppose Eq.~\eqref{eq-cond1} is violated. Then, there must exist a wedge $V(y^a)\geq0$ such that $\Theta^-_{V(y^a)}>0$. Let $V_\lambda(y^a)$ denote a one-parameter family satisfying $\partial_\lambda V_\lambda(y^a)\geq0$ for all $y^a$, and such that $V_{\lambda=0}(y^a)=0$, and $V_{\lambda=1}(y^a)=V(y^a)$. Then, by the genericity assumption~\ref{eq-generic}, there must exist a $0<\lambda^*<1$ at which Eq.~\eqref{eq-rQFC} is violated. Therefore, Eq.~\eqref{eq-rQFC} implies Eq.~\eqref{eq-cond1}.

Intuitively speaking, the genericity assumption is there to preclude scenarios like $\Theta_\lambda = \lambda^3$ which satisfy Eq.~\eqref{eq-rQFC} but violate Eq.~\eqref{eq-cond1} when going from negative to positive $\lambda$. It is highly plausible that with small additional steps along the lines of~\cite{Shahbazi-Moghaddam:2022hbw}, the genericity assumption, or condition \ref{eq-condcond} directly can be proven.\footnote{For evidence of the genericity assumption, see the discussion above Eq. (64) in~\cite{Shahbazi-Moghaddam:2022hbw}.} Alternatively, one can perhaps argue that scenarios like $\Theta_\lambda = \lambda^3$ which evade rQFC violation but violate Eq.~\eqref{eq-cond1} can be ruled out because, otherwise a small perturbation would cause a violation of the rQFC.

\subsection{Non-perturbative proof of Lemma~\ref{eq-tch}}\label{sub2}

Let us now review another gravitational constraint which was proved non-perturbatively at the brane-world species scale in~\cite{Shahbazi-Moghaddam:2022hbw}. Again, subject to an analogous genericity assumption, we show that it reduces Lemma~\ref{eq-tch} to the case where $\eth \mathcal{W}$ and $\eth \mathcal{\tilde{W}}$ coincide in a neighborhood of $p$. From here, a strong subadditivity argument was used in~\cite{C:2013uza} to complete the Lemma. In the context of a theory with a geometric UV scale, this strong subadditivity condition was conjectured to be exact in~\cite{Bousso:2024iry} (Conjecture (23) in the reference).

Let us first state the cross-focusing condition proved in~\cite{Shahbazi-Moghaddam:2022hbw} to hold exactly in the species scale. Let $\mathcal{W}$ be a spacetime wedge, which is smooth in a neighborhood of a point $p\in\eth\mathcal{W}$. Let us define coordinates $(u,v,y^a)$ at $p$, where $\partial_v$ are past-directed affine generators of $L^-$, the boundary of the past of $\mathcal{W}$, $\partial_u$ are future-directed affine generators of $L^+$, the boundary of the future of $\mathcal{W}$. And such that $u=v=0$ is $\eth\mathcal{W}$, and $y^a$ denote the transverse directions.

\begin{cond}\label{eq-2v5gs}
   Let wedges denoted by $u= U_\alpha(y^a)\geq0$ for parameter $\alpha\geq0$ satisfy 
\begin{align}
    &\partial_\alpha U_\alpha(y^a)\geq0,~~\forall \alpha,y^a\\
    &U_\alpha(y_0^a) = 0~~~\forall \alpha.
\end{align}
where $y^a_0$ denotes $p$. In words, the edge of the wedge moves to the future along $L^+$ or stay in place as $\alpha$ increases, except at $p$ where it stay in place.

The cross-focusing condition states:
\begin{align}\label{eq-crossrQFC}
    \partial_\alpha \Theta^-_{\alpha} (y_0^a)\leq 0,~~~\text{for $\alpha$ such that $\Theta^-_{\alpha}(y_0^a)=0$,}
\end{align} 
\end{cond}

In the context of Lemma~\eqref{eq-tch}, let $\mathcal{W}_b$ be a wedge satisfying $\mathcal{\tilde{W}}\subseteq \mathcal{W}_b \subseteq \mathcal{W}$ such that $\eth \mathcal{W}_b$ coincides with $\eth \mathcal{\tilde{W}}$ in a small neighborhood of $p$. Based the coordinates $(u,v,y^a)$ described above around $\eth \mathcal{W}_b$ (so that $u=v=0$ is $\eth \mathcal{W}_b$). Then $\eth \mathcal{W}$ can be specified as $u= U(y^a)\geq0$ and $v=V(y^a)\geq0$. Let us first discuss the special case where $V(y^a)=0$. We can then find a one-parameter family $U_\alpha(y^a)$ such that at $\alpha=0$ it coincides with $\eth \mathcal{W}_b$ and at $\alpha=1$ it coincides with $\eth \mathcal{W}$. Then, subject to the analogue of the genericity assumption~\eqref{eq-generic}, Eq.~\eqref{eq-crossrQFC} implies that if $\Theta^-\rvert_{\mathcal{W}}>0$, then $\Theta^-\rvert_{\mathcal{W}_b}>0$.

This argument can be generalized to the case of $V(y^a)\geq0$ by a combined use of Eq.~\eqref{eq-crossrQFC} and Eq.~\eqref{eq-rQFC}. The conclusion is that we can reduce the Lemma~\eqref{eq-tch} to the following statement. If $\mathcal{W}_b$ and $\mathcal{\tilde{W}}$ are such that $\mathcal{\tilde{W}} \subseteq \mathcal{W}_b$, and  $\eth \mathcal{W}_b$ and $\eth \mathcal{\tilde{W}}$ coincide in a neighborhood of a point $p$, then $\Theta^-\rvert_{\mathcal{W}_b}(p) > 0 \implies \Theta^-\rvert_{\mathcal{\tilde{W}}}(p)>0$. As we discussed above, this follows from conjecture (23) of Ref.~\cite{Bousso:2024iry}. But it would be interesting to directly prove it non-perturbatively on the brane-world.

\bibliographystyle{apsrev4-1}
\bibliography{all}

\end{document}